\documentclass[12pt,reqno]{amsart}
\usepackage{physics}
\usepackage[utf8]{inputenc}
\usepackage[english]{babel}
\usepackage{amsfonts,amssymb,amsmath,mathtools,braket,amssymb,empheq,stmaryrd,dsfont,physics}
\usepackage{enumerate,enumitem}
\usepackage[table, dvipsnames]{xcolor}
\usepackage{tikz}

\newlength\runit
\newcommand\numberthis{\addtocounter{equation}{1}\tag{\theequation}}

\theoremstyle{plain}
\newtheorem{theorem}{Theorem}[section]
\newtheorem{lemma}[theorem]{Lemma}

\newtheorem{proposition}[theorem]{Proposition}

\theoremstyle{definition}
\newtheorem{remark}{Remark}[section]

\renewcommand{\Re}{\operatorname{Re}}

\renewcommand{\leq}{\leqslant}
\renewcommand{\geq}{\geqslant}

\newcommand{\N}{\mathds{N}} 
\newcommand{\C}{\mathds{C}} 
\newcommand{\R}{\mathds{R}} 

\newcommand{\cF}{\mathcal{F}}
\newcommand{\cH}{\mathcal{H}}

\newcommand{\eps}{\epsilon}
\newcommand{\cA}{A}
\newcommand{\cB}{B}

\newcommand\smallO{
  \mathchoice
    {{\scriptstyle\mathcal{O}}}
    {{\scriptstyle\mathcal{O}}}
    {{\scriptscriptstyle\mathcal{O}}}
    {\scalebox{.7}{$\scriptscriptstyle\mathcal{O}$}}
  }
  \newcommand{\Par}[1]{\left( #1 \right)}
  
\newcommand{\ee}[1]{\mathrm{e}^{#1}}
\newcommand{\pscal}[1]{\ensuremath{\left\langle #1 \right\rangle}} 
\newcommand{\normt}[1]{\left\vert\kern-0.25ex\left\vert\kern-0.25ex\left\vert #1 \right\vert\kern-0.25ex\right\vert\kern-0.25ex\right\vert}

\newcommand{\di}{\,\mathrm{d}}

\DeclareMathOperator{\Id}{\mathds{1}}
\DeclareMathOperator{\sgn}{sign}

\newcommand{\ess}{\mathrm{ess}}
\newcommand{\bigO}{\mathcal{O}}

\usepackage{float}
\usepackage{graphicx}
\usepackage{subcaption}	
\usepackage{hyperref}

\title[Dirac operators in the weak coupling limit.]{Asymptotics for eigenvalues of one-dimensional Dirac operators in the weak coupling limit.}

\author[D.~Aldunate]{Danko Aldunate}
\address{Instituto de F\'isica, Pontificia Universidad Cat\'olica de Chile, Vicu\~na Mackenna 4860, Santiago 7820436, Chile.}
\email{dmaldunate@uc.cl}

\author[J.~M.~González--Brantes]{Juan Manuel González--Brantes}
\address{Instituto de F\'isica, Pontificia Universidad Cat\'olica de Chile, Vicu\~na Mackenna 4860, Santiago 7820436, Chile.}
\email{jmgonzalez4@uc.cl}

\author[H.~Van Den Bosch]{Hanne Van Den Bosch}
\address{Departamento de Ingenier\'ia Matem\'atica and Centro
de Modelamiento Matem\'atico (CNRS IRL 2807), Universidad de Chile, Beauchef 851, Santiago, Chile}
\email{hvdbosch@dim.uchile.cl}

\date{\today}

\begin{document}

\begin{abstract} 
In this paper, we derive new results on the asymptotic behavior of eigenvalues of perturbed one-dimensional massive Dirac operators in the weak coupling limit. Two classes of potentials are considered. For bounded Hermitian potentials~$V$ satisfying $|V(x)| \lesssim |x|^{-1}$ for large $|x|$, we recover the leading term, which may include a logarithmic correction if $V(x) \sim |x|^{-1}$ at infinity. 
For possibly non-Hermitian $L^1$ potentials satisfying a suitable moment condition, we obtain the second term in the asymptotic expansion.
The first result is based on a min–max principle adapted to the non-relativistic limit, while the second result is obtained via the Birman–Schwinger principle and resolvent expansions. 
\end{abstract}

\maketitle
\tableofcontents 

\section{Introduction}
\label{sec1:intro}
We consider Dirac operators of the form~$H_\epsilon=D_m - \epsilon V$, where $D_m$ is the free one-dimensional Dirac operator given by the differential expression
\begin{equation}\label{def_dirac}
    D_m  =  i \sigma_2 \partial_x  + m \sigma_3 =  \mqty(m  & \partial_{x} \\ -\partial_{x} & -m )
\end{equation}
 and $\sigma_j$ the standard $2\times 2$ Pauli matrices. The potential $V \in L^1_{\mathrm{loc}}( \R , \C^{2 \times 2})$ is decaying at infinity. We defer the precise assumptions on $V$ in each of the statements below. In any case, $\sigma_\ess(H_\epsilon) = \sigma(D_m) = \R \setminus (-m,m)$ and we are interested in the behavior of potential discrete eigenvalues of $H_\epsilon$ in the limit where the \emph{coupling constant} $\epsilon$ tends to zero. 

\bigskip

The analogous question for Schr\"odinger operators has been studied since the 70's \cite{BarrySimon1976, Klaus1977one_dimension,BlankenbeclerGoldbergerSimon1977}. In dimensions $1$ and $2$, arbitrarily small potentials $V \in L^p$ can create eigenvalues of the Schr\"odinger operator~$\frac{-\Delta}{2m} - V$. The weak coupling limit quantifies this behavior by considering a coupling constant $\epsilon >0$ and studying the asymptotics of the first eigenvalue $\lambda_V(\epsilon)$ of $-\frac{\Delta}{2m} -\eps V $ as $\epsilon$ approaches zero. We quote the part of the results that we will need, where we take into account the factor $1/(2m)$ to make the comparison with the Dirac case easier. Throughout, $\lambda^{(S)}_V(\epsilon) \le 0$ will denote the first eigenvalue of the Schr\"{o}dinger operator $-\frac{\Delta}{2m} -\eps V $ if it exists, with the convention~$\lambda^{(S)}_V(\epsilon) = 0$ if no eigenvalues exist.

\begin{theorem}[Short-range case, \cite{BlankenbeclerGoldbergerSimon1977}, Theorems 3.1 and 3.2]
\label{Sch_short_range_asymptotics}Assume that $\nu \in (0,1]$, $V:\R \to \R$ is such that $(1 + |x|^\nu)V(x) \in L^1(\R)$, and $\int V \ge 0$. Then $\frac{-\Delta}{2m} -\epsilon V$ has at least one eigenvalue for all sufficiently small $\epsilon > 0$. The smallest eigenvalue satisfies
\begin{equation} \label{eq:short-range-one-term-asymptotics}
      \lambda^{(S)}_V(\epsilon) = - \frac{1}{2m}\left( -{ m\eps} \int V  + \bigO(\eps^{1+\nu})\right)^2.
\end{equation}
If these hypotheses hold with $\nu= 1$, there is exactly one eigenvalue and
\begin{equation} \label{eq:short-range-two-term_asymptotics}
      \lambda^{(S)}_V(\epsilon) = - \frac{1}{2m}\left( -{ m\eps} \int V  - \frac{m^2\eps^2}{4} \iint V(x)|x-y|V(y)\di x \di y  + \smallO(\eps^2) \right)^2.
\end{equation}
\end{theorem}

The case where $V$ decays at infinity as~$|x|^{-1}$, i.e., the Coulomb potential, is not covered by the previous theorem. However, for bounded potentials that decay like the Coulomb potential plus short-range corrections, a detailed asymptotic expansion was derived in \cite{avronherbstsimon3}, and we reproduce its leading term below.

\begin{theorem}[Long-range-case, \cite{avronherbstsimon3} Theorem 2.1]\label{Sch_long_range_asymptotics}
    Assume that $|V(x)| \le C_1(1 + |x|)^{-1}$ and $|V(x) - 1/(1+ |x|)| \le C_2(1 + |x|)^{-1-\nu}$ for some $\nu >0$. Then $\frac{-\Delta}{2m} -\epsilon V$ has at least one eigenvalue for all sufficiently small $\epsilon > 0 $. The smallest eigenvalue satisfies
\begin{equation} \label{eq:long-range-one-term-asymptotics}
      \lambda^{(S)}_V(\epsilon) = - \frac{1}{2m} (2 m \eps)^2 |\log(\eps)|\bigl( |\log(\eps)|  + \smallO(\log(\eps))\bigr) .
\end{equation} 
\end{theorem}

Since these seminal works, weak coupling asymptotics and the phenomenon of virtual bound states (the presence of an eigenvalue for arbitrarily small couplings) have been a recurring topic in mathematical physics literature, with extensions to waveguides \cite{ExnerSeba1989waveguides, Exner1993waveguides, bulla1997waveguide}, magnetic Schr\"odinger operators \cite{vugalter1995extremal, frank2011weakly, fialova2025virtual}, periodic Schr\"{o}dinger operators \cite{zelenko2016virtual}, higher order operators \cite{arazy2006virtual}, and generalized kinetic energies with degenerate symbols \cite{weidl1999virtual, FrankHainzlNabokoSeiringer, hainzl2010asymptotic, cueninmerz2021weak, hoang2023quantitative}. For Dirac operators, the analogous question was treated first by Siegl and Cuenin in \cite{CueninSiegl2018}, who obtained the first term in the expansion for Dirac operators with potentials $V \in L^1 \cap L^2$, where $V$ is allowed to be non-Hermitian. The complementary question of the absence of eigenvalues for all sufficiently small potentials has been treated in~\cite{cossetti2020absence} and bounds for eigenvalues in dimension $d \ge 2$ have been obtained in~\cite{Dancona}.

\medskip

Here, we extend these results in complementary directions: we obtain the first term in the asymptotic expansion for the long-range case and the second term in the expansion for the short-range case. First, we treat the case of long-range potentials. We use the notation
$$
V = \begin{pmatrix}
    V_{11} & V_{12} \\ V_{21} & V_{22}
\end{pmatrix}, \quad \text{and} \quad U =\int_{\R} V\di x = \begin{pmatrix}
    U_{11} & U_{12} \\ U_{21} & U_{22}
\end{pmatrix}.
$$
If $V$ is Hermitian, then $V_{11}, V_{22} \in \R$ and $V_{21} = V_{12}^{*}$. Moreover,  for such potentials, the operator $H_{\eps}$ is unitarily equivalent to $U H_{\eps} U^{*}=D_m -\eps \tilde{V}$, with
\[U=\exp\Par{-i\eps \int_{0}^{x}  \operatorname{Im} \Par{V_{12} }(t)\di t}\, , \]
and where $\tilde{V}_{12}$ and $\tilde{V}_{21}$ are real-valued. Thus, we may assume without loss of generality that $V_{12}$ is real-valued, since its imaginary part plays the role of a magnetic potential that can be gauged away in dimension~$1$.

\begin{theorem} \label{thm:long_range}
 Assume that $V \in L^\infty(\R, \C^{2\times 2})$ is Hermitian, that the upper left component satisfies 
 $$|V_{11}(x)| \le \frac{C_1}{1 + |x|} \quad\text{ and } \quad \left|V_{11}(x) - \frac{1}{ 1+|x|}\right| \le \frac{C_2}{(1 + |x|)^{1+\nu}},$$ 
 for some $\nu >0$, and finally that $(\Re(V_{12}))^2 \le C_3(1 + |x|)^{-1}$. Then, for all sufficiently small $\eps >0$,  $D_m -\epsilon V$ has at least one positive eigenvalue~$\lambda_V^{(D)}(\epsilon)$ which satisfies
 \begin{equation} \label{eq:long-range-result}
  \lambda_V^{(D)}(\epsilon) = m -2m\eps^2 \log(\eps)^2+\smallO(\eps^{2}\log(\eps)^2).
 \end{equation}
\end{theorem}

\begin{remark}
    If, in Theorem~\ref{thm:long_range}, we replace the decay assumptions on $V_{11}$ by the corresponding assumptions on $V_{22}$, namely
$$
|V_{22}(x)| \le \frac{C_1}{1 + |x|}
\quad and \quad
\left|V_{22}(x) - \frac{1}{1+|x|}\right|
\le \frac{C_2}{(1 + |x|)^{1+\nu}},
$$
while keeping all other assumptions unchanged, then~$D_m -\eps V$ has at least one negative eigenvalue for all sufficiently small $\eps >0$. Moreover, the largest negative eigenvalue admits an asymptotic expansion of the form
    $$
    \lambda = -m + 2m\eps^2 \log(\eps)^2-\smallO(\eps^{2}\log(\eps)^2).
    $$ 
    This can be seen by noting that $\sigma_1 D_m \sigma_1 = - D_m$, and $$\sigma_1 V \sigma_1 = \begin{pmatrix}  V_{22} & V_{12} \\ V_{12}^* & V_{11} \end{pmatrix}.$$
\end{remark}

Theorem~\ref{thm:long_range} is proven in Section~\ref{sec:long_range_proof} using a min-max principle for operators with a gap as in \cite{esteban2007,SchSolTok-19} adapted to the nonrelativistic limit. It was used for the first time by two of us in collaboration with Edgardo Stockmeyer and Julien Ricaud in \cite{AldunateRicaudStockmeyerVanDenBosch2023}. To make the paper self-contained we introduce it in Section~\ref{sec:long_range_proof}. It gives a streamlined method to relate the asymptotics of Dirac operators in the weak-coupling limit with the corresponding Schr\"odinger case. Actually, Theorem~\ref{thm:long_range} is a corollary of Theorem~\ref{Sch_long_range_asymptotics} and the following more general result. 
\begin{proposition}\label{prop:general_asymptotic}
Assume that $V \in L^\infty(\R, \C^{2\times 2})$ 
\begin{enumerate}
\item is Hermitian
    \item has upper left entry $V_{11}$ satisfying the hypotheses of Theorem~\ref{Sch_long_range_asymptotics} or the hypotheses of Theorem~\ref{Sch_short_range_asymptotics} with $\nu = 1$ and $\int V_{11} >0$.
    \item has off-diagonal entries satisfying  $\Par{\Re V_{12}}^2 \lesssim \abs{V_{11}}$.
\end{enumerate}
     Then, for all sufficiently small~$\epsilon>0$, the operator $D_m -\eps V $ has at least one positive eigenvalue $\lambda_V^{(D)}(\epsilon)$ satisfying
      \begin{equation} \label{eq:comparison-result}
 \lambda_V^{(D)}(\epsilon) = m + (1 + \bigO(\epsilon^{1/2}))\lambda_{V_{11}}^{(S)}(\epsilon)  + \bigO(\epsilon^3).
 \end{equation}

 \end{proposition}

For short-range potentials this proposition, combined with Theorem~\ref{Sch_short_range_asymptotics}, allows to obtain 
 \begin{equation} 
 \lambda_{V}^{D}(\eps)=m-\frac{m}{2}U_{11}^2 \eps^2  + \smallO(\eps^2).
  \end{equation}
Hence, we recover the result of \cite{CueninSiegl2018}, but for Hermitian and bounded potentials, by a conceptually different method. It also shows that the error terms in \eqref{eq:long-range-result} can be improved by including more terms from~\cite{avronherbstsimon3} in \eqref{eq:long-range-one-term-asymptotics}.

\begin{remark} \label{rem:withoutV12}
    If $\Re(V_{12}) = 0$, \eqref{eq:comparison-result} can be refined to
     \begin{equation*}
     \lambda_V^{(D)}(\epsilon) = m + (1 + \bigO(\epsilon))\lambda_{V_{11}}^{(S)}(\epsilon).  
 \end{equation*}
 See Remark~\ref{rem:no-off-diagonal} after the proof.
\end{remark}

\medskip
 The second result in this paper concerns a two-term asymptotic formula for eigenvalues near the thresholds of Dirac operators with more rapidly decaying potentials, extending the analysis carried out in~\cite{CueninSiegl2018}. Here, our goal is to provide a refined analysis that captures the second-order asymptotics of eigenvalues bifurcating from the thresholds.
Here, we no longer assume that $V$ is Hermitian but we do need $(1+ |x|^2)V \in L^1$.
Before moving on, we recall that $V \in L^1$ suffices to define a closed extension of $D_m - \eps V$ with domain in $H^{1/2}$ and such that a Birman--Schwinger principle holds, see \cite{nenciu1976selfadjointness}.

In order to give the second-order term in the expansion, we use the boundedness of $\int |x|V(x)$ to define the following matrix in $\C^{2\times 2}$:
\begin{align*}
    F^{(\pm)}  &:= \iint V(x)\Par{\sgn(x-y) (i \sigma_2)-\frac{1}{2} \abs{x-y}\Upsilon_\pm} V(y)\di y\di x \numberthis \label{eq:definition U, F},
\end{align*}
where
\begin{align*}
    \Upsilon_\pm:=
    m(\sigma_3 \pm \Id).
\end{align*}
\begin{theorem} 
\label{thm:second_order}
Assume that $V\in L^1 $ is such that 
\[\int_\R(1+\abs{x}^2)\abs{V(x)}<\infty \, .\] As before, we define $H_\epsilon = D_m - \epsilon V$.
\begin{enumerate}
\item \label{Ev1} If $\Re(U_{11})>0$ (or $\Re(U_{11})= 0$ with $\Re(F^{+}_{11})>0 $) then, for sufficiently small $\eps >0$,  $H_\epsilon$ has an eigenvalue $z$ with expansion
    \begin{align*}
    z&= m-\frac{m}{2}U_{11}^2\eps^2 +m U_{11}F_{11}^{(+)}\eps^3 +\bigO(\eps^4)\, ,
\end{align*}
and this is the only eigenvalue of $H_\epsilon$ in the halfplane $\Re(z) > 0$. 

\item \label{Ev2} If $\Re(U_{22})< 0$, (or $\Re(U_{22}) =  0$ with $\Re(F^{-}_{22}) < 0 $)  then, for sufficiently small $\eps> 0$, $H_\epsilon$ has an eigenvalue $z$ with expansion
\begin{align*}
    z&=-m+\frac{m}{2}U_{22}^2\eps^2 -m U_{22}F_{22}^{(-)}\eps^3 +\bigO(\eps^4)\, ,
\end{align*}
and this is the only eigenvalue of $H_\epsilon$ in the halfplane $\Re(z) < 0$.

\end{enumerate}
\end{theorem}
The $\eps^2$ term in the asymptotics was found in \cite{CueninSiegl2018}, with an error term $\smallO (\epsilon^2)$, assuming $V \in L^1 \cap L^2$. To bound the error terms after extracting the first two orders in the expansion by $\bigO(\eps^4)$, we need the second moment of $V$, i.e., $(1+|x|^{2})V(x) \in L^1$. Analogously, if we only require the first moment to be finite, we can obtain an error term of the form $\smallO(\eps^3)$, see Remark~\ref{rmk:without-2nd-moment} at the end of the paper.
Finally, the hypothesis $V \in L^2$ from \cite{CueninSiegl2018} leads to some technical simplifications but is not necessary for the core of the argument, as the authors of that paper point out as well. 

\begin{remark}
For real-valued potentials, the $\eps^2$-term in Theorem~\ref{thm:second_order} coincides with the corresponding term for Schr\"odinger operators, see Theorem~\ref{Sch_short_range_asymptotics}. By going one term further in the asymptotic expansion it possibly captures \emph{relativistic corrections}.
Indeed, in the definition of $F^{(\pm)}$ in \eqref{eq:definition U, F}, the first term has no analogue in the Schr\"odinger case, whereas the second term is analogous to the next-to-leading order term appearing in~\eqref{eq:short-range-two-term_asymptotics}. However, this first relativistic term vanishes for diagonal potentials, i.e., when $V_{12} = V_{21} = 0$.  
\end{remark}

\begin{remark}
    If $V$ decays exponentially, the resolvent can be continued analytically to a strip around the branch cut. 
     For this case, the proof of Theorem~\ref{thm:second_order} also shows that, if $\Re(U_{11})< 0$ (resp. $\Re(U_{22})>0$), near the threshold $\pm m$ there exists a resonance $z_+$ (resp. $z_-$), i.e.,  a pole on the unphysical sheet for the resolvent.  It satisfies
\begin{align*} z_+ &= m+\frac{m}{2}U_{11}^2\eps^2 -m U_{11}F_{11}^{(+)}\eps^3 +\bigO(\eps^4)\, , \\
    z_-&=-m-\frac{m}{2}U_{22}^2\eps^2 +m U_{22}F_{22}^{(-)}\eps^3 +\bigO(\eps^4)\, .
\end{align*}

\end{remark}


The proof of Theorem~\ref{thm:second_order} can be found in Section~\ref{sec:short-range}. It follows a common strategy  for this type of result, see e.g. \cite{BarrySimon1976} or \cite{CueninSiegl2018} for the Dirac case. 
 The starting point is the characterization of eigenvalues by a Birman–Schwinger principle and the result then follows from a careful expansion of the Dirac resolvent. 
 

\section*{Acknowledgments}
\noindent \thanks{We are grateful to Jean-Claude Cuenin for bringing this problem to our attention and for interesting discussions about the methods. This work has been supported by ANID (Chile) through FONDECYT Projects \#125-0596 and \#123-1539 and the Center for Mathematical Modeling through Basal Funding \#FB21-0005. DA and JMGB thank the Center of Mathematical Modeling at the Universidad de Chile for its hospitality during this research. We are indebted to the anonymous referees for their meticulous review of the preprint; their suggestions greatly enhanced the clarity and quality of this paper. }

\section{The min-max principle and the long range case.} 
\label{sec:long_range_proof}

Our starting point in this section is a min-max principle for eigenvalues inside a gap in the essential spectrum of an operator, as shown first in~\cite{DolEstSer-00,GriSei-99} (see e.g.,~\cite{DolEstSer-06,EstLewSer-19,SchSolTok-19} for related results). 

\begin{theorem}[Theorem 1 of~\cite{DolEstSer-06}]\label{minmax_with_gap}
    Let $A$ be a self-adjoint operator wit domain~$D(A)$, in a Hilbert space $\cH$. Suppose that $\Lambda_{\pm}$ are orthogonal projections on $\cH$ with $\Lambda_+ + \Lambda_- = \Id_\cH$ and such that 
    \[
        \mathcal{F}_\pm:=\Lambda_\pm D(A) \subset D(A).
    \]
    Define $\gamma_0$, the lower limit of the gap, as
      \begin{equation}\label{def_0_minmax_level}
        \gamma_0:=\sup\limits_{x_-\in \mathcal{F}_{-}\setminus\{0\}} \frac{\pscal{x_-,A x_-}_\cH}{\norm{x_-}_\cH^2},
    \end{equation}
    and~$\gamma_\infty$, its upper limit, as
    \[
        \gamma_\infty:=\inf(\sigma_\ess(A)\cap(\gamma_0,+\infty))\in[\gamma_0,+\infty].
    \]
    Finally, for $k \in \N\setminus\{0\}$, the min-max levels are defined as 
     \begin{equation}\label{def_minmax_levels}
		    \gamma_k := \inf\limits_{\substack{V\subset \mathcal{F}_+\\\dim V = k}}\sup\limits_{x\in (V\oplus \mathcal{F}_-)\setminus\{0\}}\frac{\pscal{x,Ax}_\cH}{\norm{x}_\cH^2}.
	\end{equation}
	If $\gamma_0 <+ \infty $ and the \emph{gap condition}
	\begin{equation}\label{gap_condition_in_thm}
		\gamma_0 < \gamma_1
	\end{equation}
	is satisfied, then for any $k\geq1$ either $\gamma_k$ is the $k$-th eigenvalue of~$A$ in~$(\gamma_0, \gamma_\infty)$, counted with multiplicity, or $\gamma_k=\gamma_\infty$. In particular, $\gamma_\infty \geq \sup_{k \geq 1}\gamma_k \geq \gamma_1$.
\end{theorem}
In this work, we follow the approach in \cite{AldunateRicaudStockmeyerVanDenBosch2023}. Thus, we set $\alpha = (2m)^{-1}$ and consider the operators
\begin{equation}\label{lambdaplus}
\Lambda_{+} =   \begin{pmatrix} 1 & + \alpha \partial_x \\ -\alpha \partial_x & - \alpha^2 \partial_x^2 \end{pmatrix} (1 -\alpha^2 \partial_x^2)^{-1} ,
\end{equation}
and
\begin{equation}\label{lambdaminus}
\Lambda_{-} =   \begin{pmatrix}- \alpha^2 \partial_x^2 & - \alpha \partial_x \\ + \alpha \partial_x & 1 \end{pmatrix} (1 -\alpha^2 \partial_x^2)^{-1} ,
\end{equation}
which are bounded since the resolvent $(1-\alpha^2 \partial_x^2)^{-1}$ is bounded from~$L^2(\R, \C)$ to $H^2(\R, \C)$. 
They satisfy
\begin{equation*}
	\Lambda_{+} + \Lambda_{-} = \Id_\mathcal{H}\,, \quad \Lambda_{\pm}\Lambda_{\mp}=0, \quad 
    \Lambda_\pm^*= \Lambda_\pm \text{ and } \Lambda_{\pm}^2 =\Lambda_{\pm}, 
\end{equation*}
so they correspond to complementary orthogonal projections in $L^2(\R, \C^2)$.
Applying $\Lambda_\pm$ to the domain of the Dirac operator yields
\begin{equation}
\label{eq:positivesubspace}
\cF_+ = \Lambda_{+} H^1(\R, \C^2) = \left\lbrace \begin{pmatrix} h \\ -\alpha h' \end{pmatrix} = l_{+}\left(h\right) \middle| h \in H^{2}\left(\R\right) \right\rbrace,
\end{equation}
\begin{equation}
\label{eq:negativesubspace}
\cF_- = \Lambda_{-} H^1(\R, \C^2) = \left\lbrace \begin{pmatrix} -\alpha g' \\ g \end{pmatrix} = l_{-}\left(g\right) \middle| g \in H^{2}\left(\R\right) \right\rbrace.
\end{equation}

To obtain the last identities, we used
\begin{align*}
    \Lambda_{+}\begin{pmatrix} f \\ g \end{pmatrix} =  \begin{pmatrix} (1-\alpha \partial_x^2)^{-1}\Par{f+\alpha \partial_x g} \\ -\alpha \partial_x(1-\alpha \partial_x^2)^{-1}\Par{f+\alpha \partial_x g} \end{pmatrix} =: \begin{pmatrix} h \\ -\alpha h' \end{pmatrix} ,
\end{align*}
with $h= (1-\alpha \partial_x^2)^{-1}\Par{f+\alpha \partial_x g} \in H^2(\R, \C)$ and its analogue for $\cF_-$. 

\bigskip

Our motivation to consider these particular subspaces lies in their potential to naturally relate the behavior of eigenvalues near the thresholds $\pm m$ to the corresponding Schrödinger eigenvalues, see the proof of Proposition~\ref{prop:general_asymptotic} below.
At an informal level, this comes from the fact that by taking a Fourier transform, $\widehat \Lambda_{\pm}(k)$ is a first-order-in-$k$ approximation of the eigenprojections of the free Dirac operator on the subspaces with energies $\pm \sqrt{m^2 + k^2}$. As opposed to these pseudodifferential eigenprojections, $\Lambda_\pm$ have the advantage of allowing an explicit expression in terms of differential operators.

\bigskip

Now, for any $x\in H^1 (\R,\C^2)$, we have the decomposition $x = l_{+}\left(h\right) + l_{-}\left(g\right)$ so that expectation values of $H_\eps = D_m - \eps V$ are given by
\begin{equation}
\label{eq:rayleighquotientnumerator}
    \left \langle x, H_\eps x \right \rangle = \left \langle l_{+}\qty(h), H_\eps l_{+}\qty(h) \right \rangle + \left \langle l_{-}\qty(g), H_\eps l_{-}\qty(g) \right \rangle + 2 \Re \left \langle l_{+}\qty(h), H_\eps l_{-}\qty(g) \right \rangle\, . 
\end{equation}
We begin by expanding the quantities that will be needed later. Using the definitions of the subspaces~\eqref{eq:positivesubspace} and~\eqref{eq:negativesubspace} (recall that $\alpha = (2m)^{-1}$) together with the definition of the Dirac operator~\eqref{def_dirac}, we obtain
\begin{align}
  \label{eq:lplus-sandwich}
	\left \langle l_{+}\qty(h), H_{\eps} l_{+}\qty(h) \right \rangle & = m \norm{l_{+}\qty(h)}^{2} + \alpha \norm{h'}^{2}  - \epsilon \pscal{l_+(h),V l_+(h)}\, , \\
    \label{eq:lminus-sandwich}
	\left \langle l_{-}\qty(g), H_{\eps} l_{-}\qty(g) \right \rangle &= -m \norm{l_{-}\qty(g)}^{2} - \alpha \norm{g'}^{2}  - \epsilon \pscal{l_-(g),V l_-(g)}\, ,\\
    \label{eq:crossedterms}
	\left \langle l_{+}\qty(h), H_{\eps} l_{-}\qty(g) \right \rangle & =\alpha^{2} \left \langle h'', g '\right \rangle - \epsilon \pscal{l_+(h),V l_-(g)}\, .
\end{align}

In order to apply the min-max principle, the following Lemma proves the \emph{gap condition}~\eqref{gap_condition_in_thm} of Theorem~\ref{minmax_with_gap}. 

\begin{lemma} \label{lem:GapCondition} As above, let $H_\eps = D_m - \eps V$ with $V\in L^\infty$. Then, for any $\epsilon < m\norm{V}_{\infty}^{-1}$, the operator~$H_\eps$~satisfies the gap condition of~Theorem~\ref{minmax_with_gap}, i.e.,~$\gamma_0 < 0 < \gamma_1$, and the upper limit~$\gamma_\infty =m$.
\end{lemma}

\begin{proof}
\noindent First we estimate $\gamma_{0}$ using \eqref{eq:lminus-sandwich}, 
\begin{align}
\gamma_{0} &= \sup_{g\neq 0} \dfrac{-m \norm{l_{-}\qty(g)}^{2} - \alpha \norm{g'}^{2} - \epsilon \pscal{l_-(g),V l_-(g)}}{\norm{l_{-}\qty(g)}^{2}}, \nonumber \\
&= -m +\sup_{g\neq 0}\dfrac{-\alpha \norm{g'}^{2} - \epsilon \pscal{l_-(g),V l_-(g)}}{\norm{l_{-}\qty(g)}^{2}}, \nonumber \\
&\leq -m - \inf_{g\neq 0} \dfrac{\alpha \norm{g'}^{2} + \epsilon \pscal{l_-(g),V l_-(g)}}{\norm{l_{-}\qty(g)}^{2}}, \nonumber \\
&\leq -m + \epsilon \norm{V}_{\infty}.
\end{align}
On the other hand, taking $g=0$,
\begin{align}
\label{eq:first-eigenvalue}
\gamma_{1}  &= \inf_{h \neq 0} \sup_{g} \dfrac{\left \langle l_{+}\qty(h), H_\eps l_{+}\qty(h) \right \rangle + \left \langle l_{-}\qty(g), H_\eps l_{-}\qty(g) \right \rangle + 2\Re \left \langle l_{+}\qty(h), H_\eps l_{-}\qty(g) \right \rangle}{\norm{l_{+}\qty(h)}^{2} + \norm{l_{-}\qty(g)}^{2}}, \nonumber \\
&\geq \inf_{h \neq 0} \dfrac{m\norm{l_{+}\qty(h)}^{2} + \alpha\norm{h'}^{2} - \epsilon \pscal{l_+(h),V l_+(h)}}{\norm{l_{+}\qty(h)}^{2}}, \nonumber \\
&= m + \inf_{h \neq 0} \dfrac{\alpha\norm{h'}^{2} - \epsilon \pscal{l_+(h),V l_+(h)}}{\norm{l_{+}\qty(h)}^{2}}, \nonumber \\
&\geq m - \epsilon \norm{V}_{\infty}.
\end{align}
Given any $\epsilon < m\norm{V}_{\infty}^{-1}$ we have $\gamma_0 <0<\gamma_1$ and the upper limit $\gamma_\infty $ is given by
    \[
        \gamma_\infty = \inf \Bigl(\bigl((-\infty,-m]\cup [m,+\infty)\bigr)\cap(\gamma_0,+\infty)\Bigr)=m\,. \qedhere
    \]
\end{proof}
This Lemma shows that, for sufficiently small values of $\epsilon>0$, $\gamma_1$ is either the lowest positive eigenvalue of $H_\epsilon$ or the bottom of the positive part of its essential spectrum. To establish that $\gamma_1$ is an eigenvalue, it is necessary that $\gamma_1 <m$. To this end, we establish matching lower and upper bounds, on $\gamma_1$ which show that $\gamma_1 < m$  and capture its leading-order-asymptotics in $\eps$, thus proving Proposition~\ref{prop:general_asymptotic}.

\begin{proof}[Proof of Proposition~\ref{prop:general_asymptotic}]
Recall that we may assume $V_{12}$ is real valued and that $\gamma_\infty = m$.
From this point on, we will take the convention that $\lambda_V^{(D)} (\epsilon)$ is the smallest element of the spectrum of $D_m -\epsilon V$ in $(0, +\infty)$, such that it is the smallest positive eigenvalue if such an eigenvalue exists, or the upper threshold $m$ else. For the Schrödinger case, we use the analogous convention and write $$\lambda_V^{(S)}(\eps) = \inf \left(\operatorname{spec}(-\alpha \Delta - \eps V) \right).$$ 

For instance, with this definition, whenever $\gamma_1 >\gamma_0$ we have
$$
\lambda_V^{(D)} (\epsilon) = \gamma_1.
$$

We split the proof into two parts: one addressing the lower bounds on $\gamma_1$, and the other the upper bounds.

\medskip

\noindent \textbf{Lower bounds.} To get lower bounds on $\gamma_{1}$, we take $g = 0$. Then,
\begin{align}
\gamma_{1} &\geq \inf_{h \neq 0} \dfrac{m\norm{l_{+}\qty(h)}^{2} + \alpha\norm{h'}^{2} - \epsilon  \pscal{l_+(h), V l_+(h)}}{\norm{l_{+}\qty(h)}^{2}}, \nonumber \\
&= m + \inf_{h \neq 0} \dfrac{\alpha \norm{h'}^{2} - \epsilon \int V_{11} \abs{h}^{2} -\eps \int \alpha^{2} V_{22} \abs{h'}^{2}
+ 2 \eps\int \alpha V_{12} \Re (\bar h h')
}{\norm{l_{+}\qty(h)}^{2}}. 
\end{align}

To get rid of the entries of $V_{12}$ and $V_{22}$, we bound

\begin{equation} \label{eq:bound_V22}
   \left| \int  \alpha^{2}V_{22} \abs{h'}^{2}  \right| \le \alpha^2 \norm{V_{22}}_{\infty} \norm{h'}^2.
\end{equation}

and

\begin{equation} \label{eq:bound_V12}
    2 \alpha | V_{12} \Re(\bar  h h')| \le \epsilon^{1/2} |V_{12}|^2 |h|^2 + \alpha^2 \epsilon^{-1/2} |h'|^2,
\end{equation}

such that

\begin{align*}
\gamma_{1} 
&\ge m + \inf_{h \neq 0} \dfrac{\alpha \left( 1 -\epsilon^{1/2}\alpha - \epsilon \alpha \norm{V_{22}}_{\infty} \right)\norm{h'}^{2} - \epsilon \int (V_{11} +\epsilon^{1/2} V_{12}^2) \abs{h}^{2} 
}{\norm{l_{+}\qty(h)}^{2}}.
\end{align*}

To clean up the expressions, define
$$
 V_- := V_{11} +\epsilon^{1/2} V_{12}^2, \quad  \kappa_- := \left( 1 -\epsilon^{1/2}\alpha - \epsilon \norm{V_{22}}_{\infty} \right)  \text{ and } \epsilon_- = \epsilon / \kappa_-.
$$
Note that for $\eps$ sufficiently small, $\kappa_- > 0$. We will now bound~$\gamma_1$ in terms of the lowest eigenvalue of the Schr\"odinger operator~$-\alpha \Delta - \eps_- V_-$. Thus,
\begin{align*}   
\gamma_{1}  
&\ge m + \kappa_- \inf_{h \neq 0} \dfrac{\alpha \norm{h'}^{2} - \epsilon_- \int V_- \abs{h}^{2} 
}{\norm{l_{+}\qty(h)}^{2}}, \\
&\ge m  + \kappa_- \lambda_{V_-}^{(S)}(\epsilon_-) \sup_{h \neq 0} \frac{\norm{h}^2}{\norm{l_+(h)}^2}, \\
& =  m  + \kappa_- \lambda_{V_-}^{(S)}(\epsilon_-).
\end{align*}
Our hypothesis on $V_{12}$ implies that $V_-$ satisfies either the hypothesis of Theorem~\ref{Sch_long_range_asymptotics} (with constants modified to order $\eps^{1/2}$), or of Theorem~\ref{Sch_short_range_asymptotics} where $\int V_- > 0$ for sufficiently small $\eps > 0$. In either case, both eigenvalues $\lambda^{(S)}_{V_-}(\eps)$ and $\lambda^{(S)}_{V_{11}}(\eps)$ have asymptotic expansions of the same leading order for small $\epsilon$, such that
\[\lambda_{V_-}^{(S)}(\epsilon_-) \ge  (1 + \bigO(\epsilon^{1/2})) \lambda_{V_{11}}^{(S)}(\epsilon).\]
We obtain the lower bound 
\begin{equation} \label{eq:lower_bound_gamma}
    \gamma_1 \ge m -(1 + \bigO(\epsilon^{1/2}))\lambda_{V_{11}}^{(S)}(\epsilon).
\end{equation}

\noindent \textbf{Upper bounds.}

Recall that
\begin{align*}
    \gamma_{1} 
    &= \inf_{h \neq 0} \sup_{g} \dfrac{\left \langle l_{+}\qty(h), H_{\eps}l_{+}\qty(h) \right \rangle + \left \langle l_{-}\qty(g), H_{\eps} l_{-}\qty(g) \right \rangle + 2 \Re \left \langle l_{+}\qty(h), H_{\eps} l_{-}\qty(g) \right \rangle}{\norm{l_{+}\qty(h)}^{2} + \norm{l_{-}\qty(g)}^{2}}
\end{align*}
We analyze the terms in the numerator separately.
Using the definition of $\gamma_0$, we may simply bound
\begin{equation}
  \pscal{ l_{-}\qty(g), H_{\eps} l_{-}\qty(g) } \le \gamma_0 \norm{l_-(g)}^2.
\end{equation}
Next, using \eqref{eq:lplus-sandwich} and bounds \eqref{eq:bound_V12} and \eqref{eq:bound_V22}, we have,
\begin{align*}
    \left \langle l_{+}\qty(h), H_{\eps}l_{+}\qty(h) \right \rangle 
    &= m\norm{l_{+}\qty(h)}^{2} + \alpha\norm{h'}^{2} - \epsilon  \pscal{l_+(h), V l_+(h)}, \\
    &\le m \norm{l_{+}\qty(h)}^{2} +  \alpha \left( 1 +\epsilon^{1/2}\alpha + \epsilon \alpha \norm{V}_{\infty} \right)\norm{h'}^{2} \\
    & \qquad \quad - \epsilon \int (V_{11} - \epsilon^{1/2} V_{12}^2) \abs{h}^{2}    .
\end{align*}
For the cross-terms, we expand \eqref{eq:crossedterms} to obtain 
\begin{align*}
   \pscal{l_+(h), H_\epsilon l_-(g) } = \pscal{\alpha h'' +\epsilon V_{11} h, \alpha g'} - \epsilon \pscal{h, V_{12}g} - \epsilon \alpha \pscal{h',  \alpha V_{12} g' - V_{22} g}.
\end{align*}
Using a bound similar to \eqref{eq:bound_V12} for the second term above and standard Cauchy-Schwarz for the others, yields
\begin{align*}
    2| \Re \pscal{l_+(h), H_\epsilon l_-(g)} | &\le \epsilon^{3/2} \pscal{h, V_{12}^2 h} + \epsilon^{1/2}\norm{l_-(g)}^2 \\
    & \quad \quad + 2\norm{-\alpha h'' -
 \epsilon V_{11} h}\norm{l_-(g)} \\
 & \qquad \qquad + 4\epsilon \alpha \norm{V}_{\infty} \norm{h'}\norm{l_-(g)}.
\end{align*}
This suggests to define
$$
V_+ := V_{11} - 2\epsilon^{1/2} V_{12}^2, \quad \kappa_+ := \left( 1 +\epsilon^{1/2}\alpha + \epsilon \alpha \norm{V}_{\infty} \right), \quad \epsilon_+ := \epsilon/\kappa_+,
$$
such that
\begin{align*}
  &  \pscal{l_+(h) + l_-(g), H_\epsilon (l_+(h) + l_-(g))} \\
  &  \qquad \le m \norm{l_{+}(h)}^{2} + (\gamma_0 + \epsilon^{1/2} ) \norm{l_-(g)}^2 + \kappa_+ \left( \alpha \norm{h'}^2 -\epsilon_+ \pscal{h, V_+ h}\right) \\
 & \quad \qquad + 2 \norm{l_-(g)} \left( \norm{-\alpha h'' - \epsilon V_{11} h} + 2\epsilon {\alpha} \norm{V}_{\infty} \norm{h'} \right). \numberthis \label{eq:upper-bound-cross-term}
  \end{align*}
For all $\epsilon$ sufficiently small,
$\lambda^{(S)}_{V_+}(\epsilon_+)$ is well-defined, and we take $h = h_\epsilon$, the associated normalized eigenfunction. We will nevertheless write~$\norm{h_\eps}$ to keep track of all the terms.
This gives
\begin{align*}
     \alpha \norm{h'_{\epsilon}}^2 -\epsilon_{+} \pscal{h_{\epsilon}, V_+ h_{\epsilon}} &= \lambda^{(S)}_{V_+}(\epsilon_+) \norm{h_\epsilon}^2 
     \end{align*}
 for the main term. For the cross term, we use the bounds    
 \begin{align*}
     \norm{-\alpha h_{\epsilon}'' - \epsilon V_{11} h_{\epsilon}} & \le \left(\lambda^{(S)}_{V_+}(\epsilon_+) + \epsilon \norm{V_{11} - V_+}_{\infty}\right)  \norm{h_\epsilon}, \\
      \alpha \norm{h'_{\epsilon}}^2 &\le \epsilon_+ \pscal{h_\epsilon, V_+ h_\epsilon } \le \epsilon_+ \norm{V_+}_{\infty}  \norm{h_\epsilon}^2 ,
\end{align*}
so we can bound the parenthesis in $\eqref{eq:upper-bound-cross-term}$ by 
\begin{align*}
     & \norm{-\alpha h_\eps'' - \epsilon V_{11} h_\eps} + 2\epsilon \alpha \norm{V}_{\infty} \norm{h_\eps'}  \\
     & \leq \Par{\lambda^{(S)}_{V_+}(\epsilon_+) + \epsilon \norm{V_{11} - V_+}_{\infty}+2\eps\sqrt{\alpha\epsilon_+} \norm{V}_{\infty} \norm{V_+}_{\infty}^{1/2}}\norm{h_\eps}\\
    & = \bigO(\eps^{3/2})\norm{h_\eps},
\end{align*}
where we have used again the asymptotics of $\lambda_{V_+}^{(S)}(\eps_+)$ and $ \eps_+^2|\log \eps_+|^2 \ll \eps^{3/2} $.
\medskip
Inserting these bounds in \eqref{eq:upper-bound-cross-term}, we get,
\begin{align*}
&\pscal{l_+(h_\eps) +  l_-(g),H_\eps (l_+(h_\eps) + l_-(g))}\\ &
\quad \le  m\norm{l_+(h_\epsilon)}^2 +   \kappa_+ \lambda^{(S)}_{V_+}(\epsilon_+) \norm{h_\eps}^2 \\
& \quad\qquad +(\gamma_0 + \epsilon^{1/2} ) \norm{l_-(g)}^2   +  \bigO(\epsilon^{3/2})\norm{l_-(g)}\norm{h_\epsilon}, \\
& \quad \le m \norm{l_+(h_\epsilon)}^2  + \kappa_+ \lambda^{(S)}_{V_+}(\epsilon_+) \norm{h_\epsilon}^2  +  \bigO(\epsilon^{3})\norm{h_\epsilon}^2,
\end{align*}
where the last inequality follows from completing the square and using $\gamma_0 + \eps^{1/2} < 0$. Thus,
\begin{align*}
\gamma_1 
&\le \sup_{g} \frac{ m \norm{l_+(h_\epsilon)}^2  + \kappa_+ \lambda^{(S)}_{V_+}(\epsilon_+) \norm{h_\epsilon}^2  +  \bigO(\epsilon^{3})\norm{h_\epsilon}^2}{\norm{l_{+}\qty(h_\epsilon)}^{2} + \norm{l_{-}\qty(g)}^{2}}.
\end{align*}
Since the numerator is positive for all sufficiently small $\epsilon$, we conclude that the supremum is at $g= 0$, so that we are left with
\begin{align*}
\gamma_1 
&\le \frac{ m \norm{l_+(h_\epsilon)}^2  +  \kappa_+ \lambda^{(S)}_{V_+}(\epsilon_+) \norm{h_\epsilon}^2  +  \bigO(\epsilon^{3})\norm{h_\epsilon}^2}{\norm{l_{+}\qty(h_\epsilon)}^{2} }, \\
& \le m  + \kappa_+ \lambda^{(S)}_{V_+}(\epsilon_+) + \bigO(\epsilon^{3}), \\
& \le m + \lambda^{(S)}_{V_{11}}(\epsilon)(1 + \bigO(\epsilon^{1/2})) + \bigO(\epsilon^{3}), 
\end{align*}
where in the last line we used once more that asymptotics of $ \lambda^{(S)}_{V_+}(\epsilon_+)$ and $\lambda^{(S)}_{V}(\epsilon)$ coincide up to order $\epsilon^{1/2}$. In particular, this shows that $\gamma_1 < m$, and that the upper and lower bounds match to leading order, thereby concluding the proof.\end{proof}

 Combining Proposition~\ref{prop:general_asymptotic} with the Schr\"{o}dinger expansion from Theorem~\ref{Sch_long_range_asymptotics} establishes the eigenvalue asymptotics for Dirac operators with long-range potentials stated in Theorem~\ref{thm:long_range}.

\begin{remark} \label{rem:no-off-diagonal}
In the case $\Re V_{12} = 0$, the proof simplifies considerably and we can apply the same strategy with $V_+ = V_- = V_{11}$, and $|\kappa_{\pm }-1|\le C \eps$. 
\end{remark}
\section{Second-order expansion and the proof of Theorem~\ref{thm:second_order}} \label{sec:short-range}
In this section, we prove Theorem~\ref{thm:second_order}. For the potential, we write $V = B^* A$ (with, for instance, $\cB^*:=U_{V}\abs{V}^{\frac{1}{2}}$, $\cA:=\abs{V}^{\frac{1}{2}}$ and $V = U_V |V|$ its polar decomposition). Since $V \in L^1$, the components of $A,B$ belong to $L^2$.

\medskip

The Birman--Schwinger principle then gives that $z$ is an eigenvalue of $D_m - \eps V$ if and only if $1$ is an eigenvalue of $K(z)= \epsilon A (D_m -z)^{-1} B^*$. It holds for any $V \in L^1$ by using the closed extension of $D_m - \eps V$ \emph{generated} from the perturbed resolvent following the strategy in \cite{nenciu1976selfadjointness} (see, e.g., \cite[Proposition 2.1]{dolbeault2023keller}).
Recall that the kernel of the one-dimensional Dirac resolvent $(D_m - z)^{-1}$ is
\begin{align}
    \label{eq:def_resolvent}
     R_z(x,y):=\frac{e^{- \kappa (z)|x- y|}}{2} \begin{pmatrix}  (z+m)/\kappa (z) & \sgn (x-y) \\ -\sgn (x-y) & (z-m)/\kappa (z) 
     \end{pmatrix}, 
\end{align} 
where $\kappa (z)=\sqrt{m^2-z^2}$ is chosen with $\Re(\kappa(z))\geq 0$. 
We separate the singular part near the threshold $m$, so we define $P_+ = (1, 0)$ and write
 \begin{align*} 
        R_z(x,y)
     &=
    e^{-\kappa (z) |x-y|}  \left( \frac{ m}{\kappa(z)}P_+^* P_+  - \frac{\sgn(x-y)}{2}i\sigma_2  +\frac{z-m}{2 \kappa (z)} \Id_{\C^2} \right), \\
    & =
    \frac{ m}{\kappa (z)} P_+^* P_+ + (e^{-\kappa(z)|x-y|} -1 )\frac{ m}{\kappa(z)} P_+^* P_+  \\
    & \qquad + e^{-\kappa(z)|x-y|} \left(    -\frac{\sgn(x-y)}{2}i\sigma_2  +\frac{z-m}{2 \kappa(z)} \Id_{\C^2} \right), \\
     &=:  \frac{m}{\kappa(z)} P_+^* P_+ + S_z(x,y).
\end{align*} 
This allows to decompose the Birman-Schwinger operator as
\[K(z) = \frac{\eps m}{\kappa (z)}A P_+^* P_+ B^* + \eps M(z)\,,\]
where the kernel of $M(z)$ is
\begin{equation}
    M_z(x,y) = A(x) S_z(x,y)B^*(y).
\end{equation}
We also define the vectors
\begin{align}
    b = B P_+^*, \quad a = A P_+^* \quad \text{ such that } a,b \in L^2(\R, \C^2).
\end{align}

\begin{lemma}\label{lemma_BS} Assume $V \in L^1$.
    With the previous constructions, assume that $\epsilon \norm{M(z)} < 1$. Then, $z$ is an eigenvalue of $D_m - \eps V$ if and only if
    \begin{equation} \label{eq:lemma_BS}
       \frac{\epsilon m}{\kappa(z)}\pscal{b, (1 - \eps M(z))^{-1} a}_{L^2} = 1.
    \end{equation}

\end{lemma}
\begin{proof}
By the Birman-Schwinger principle, $z$ is an eigenvalue of $D_m - \eps V$ if and only if $1$ is an eigenvalue of $K(z)$
such that we can rewrite
$$
K(z) = \frac{ \epsilon m}{\kappa(z)} | a\rangle \langle b|  + \epsilon M(z).
$$
For any fixed value of $z \in \C \setminus \bigl( (-\infty, -m] \cup [m, + \infty)\bigr)$, $M(z)$ can be easily seen to be a Hilbert-Schmidt operator since $S_z(x,y)$ is a bounded matrix and $A, B \in L^2$.
Now, we assume that $\eps \norm{M(z)}  < 1$. Then $\phi\neq 0$ is an eigenvector with eigenvalue $1$ of $K(z)$ if and only if
$$
\frac{\eps m}{\kappa(z)} a \pscal{b, \phi} = (1 - \eps M(z)) \phi  \quad \Leftrightarrow\quad  \frac{\eps m}{\kappa(z)}(1 - \eps M(z))^{-1} a    \pscal{b, \phi} =  \phi.
$$
Since $(1 - \eps M(z))$ is invertible, we notice that $\pscal{b, \phi} \neq 0 $ for any eigenvector.
Taking the inner product with $b$ to both sides we get
$$
 \frac{ \eps m}{\kappa(z)} \pscal{b, (1 - \eps M(z))^{-1} a}= 1. \qedhere
$$
\end{proof}

The next step is to expand the zeros of this function using Rouché's theorem. 
Recall that we have defined
\begin{align*}
       F^{(\pm)}  &:= \iint V(x)\Par{\sgn(x-y) (i \sigma_2)-\frac{1}{2} \abs{x-y}\Upsilon_\pm} V(y)\di y\di x 
\end{align*}
and 
\begin{align*}
    \Upsilon_\pm :=
    m(\sigma_3 \pm \Id)  = 2 m P^*_+ P_+ .
\end{align*}
\begin{lemma}
    Assume that $V$ is such that $\Re U_{11} > 0 $ (or $\Re(U_{11}) = 0$ with $\Re F^{+}_{11}>0$) and $\epsilon>0$ is sufficiently small. Then, 
    $$
\frac{\epsilon m}{\kappa(z)}\pscal{b, (1 - \eps M(z))^{-1} a}_{L^2} = 1
    $$
    has exactly one solution $z_0$ in the halfplane $ \Re(z) > 0$. 
    It satisfies 
    $$z_0 = m - \frac{1}{2m }(\eps m U_{11} + \eps^2 m F^+_{11})^2 + \bigO(\eps^4). $$
\end{lemma}
\begin{proof}
    We restrict to values of $z$ in $\Omega_+ = \{z \in \C\setminus [m, + \infty)| \Re(z) > 0\}$ and first show that it is possible to take $\eps>0$ sufficiently small such that for all  such $z$, we have $\eps \norm{M(z)} < 1$. 
    To this end, recall that we have defined
\begin{align*}
    S_z(x,y) &= (e^{-\kappa(z)|x-y|} -1 )\frac{ m}{\kappa(z)} P_+^* P_+ \\
    &\qquad + \frac{e^{-\kappa(z)|x-y|}}{2}\begin{pmatrix} (z-m)/\kappa(z) & -\sgn(x-y) \\ \sgn(x-y) & (z-m)/\kappa(z) \end{pmatrix}\, .
\end{align*}
In $\Omega_+$,  we have $|\kappa(z)| > |m-z| $, hence
\begin{align*}
    |S_z(x,y) |\le C ( m|x-y| + 1) \Id_{\C^2}  \le C ( m|x| + m|y| + 1) \Id_{\C^2} .
\end{align*}
This gives the following bound for the Hilbert-Schmidt norm of $M$:
\begin{align*}
    \norm{M(z)}_{HS} \le C \norm{A(x)(1 +|x|)}_{L^2} \norm{B(y)(1 +|y|)}_{L^2} = C\norm{V(1 + |x|)^2}_{L^1}.
\end{align*}
It shows that we can apply Lemma~\ref{lemma_BS} for all $z$ in $\Omega_+$ and all sufficiently small $\epsilon>0$. 
    It is convenient to expand in $\kappa(z) = \sqrt{m^2 - z^2}$, which is a conformal map from
    $\Omega_+ $ to itself. It has the advantage that the threshold $z \to m$ corresponds now to $\kappa \to 0$ and that the branch cut for the resolvent is mapped to the imaginary axis.
Thus, we look for zeros of the analytic function
$$
g_\epsilon: \Omega_+ \to \C, \quad g_\epsilon(\kappa) :=  \epsilon m \pscal{b, (1 - \eps M(z(\kappa)))^{-1} a}_{L^2} - \kappa.
$$
From the previous bound on the Hilbert-Schmidt norm, we conclude that $g_\epsilon$ also extends continuously to the imaginary axis.
We expand
$$
g_\eps(\kappa) =  \epsilon m \pscal{b, (1 + \eps M(z(\kappa))) a}_{L^2} - \kappa + \bigO(\epsilon^3) .
$$
The first term gives the leading order in $\epsilon$ and we check that 
$$
\pscal{b, a} = P_+ \int_\R B^*(x) A(x) \di x P_+^* = P_+ \int_\R V(x) \di x P_+^* = U_{11}
$$
so that
$$
g_\eps(\kappa) =  \eps m U_{11} - \kappa + \bigO(\eps^2).
$$
For the next-to-leading order in $\epsilon$, we need to expand the kernel of $M(z)$ to leading order in $\kappa$. 
We define $M^{(1)}$ as the operator with kernel
$$
M^{(1)}(x,y) =  A(x)\begin{pmatrix} - m |x-y| & -\sgn(x-y)/2 \\
\sgn(x-y)/2 & 0\end{pmatrix}
    B^*(y)
$$
such that
\begin{align*}
    &|M_{z}(x,y) - M^{(1)}(x,y) | \\
    &\le |A(x)|\left(\frac{\ee{-\kappa(z)|x-y|}-1+\kappa(z)\abs{x-y}}{\kappa(z)}  P_+^{*}P_{+} \right. \\
    &\qquad \qquad \left. -\frac{i}{2}\sgn(x-y)\sigma_2 \Par{\ee{-\kappa(z)|x-y|}-1}+\frac{z-m}{\kappa(z)}\Id_{\C^2} \right)  |B(y)| \\
    & \le C |\kappa(z)| (1+|x-y|^2) |A(x)| |B(y)| .
\end{align*}
The last inequality above is obtained by bounding the second term by $\kappa(z) |x-y| \lesssim \kappa(z) (1+ |x-y|^2)$ and the last term by
$\kappa(z) \sim \sqrt{m-z}$.
Therefore
\begin{align*} 
   & \left|\pscal{b,  \Par{M(z)- M^{(1)} }a }_{L^2}  \right|\\
    & \le C |\kappa| \iint|B(x)| |A(x)| (1 + |x|^2 + |y|^2)|B(y)| |A(y)| \di x \di y
    \\
    &\le C |\kappa |\norm{V(x)(1+|x|^2)}_{L^1}^2.
    \numberthis  \label{eq:snd-moment-bound}
\end{align*}
Now we compute the main term
\begin{align*}
   & \pscal{b,  M^{(1)} a}_{L^2} \\
      &= \iint P_+ B^*(x) A(x) \begin{pmatrix} - m |x-y| & -\sgn(x-y)/2 \\
\sgn(x-y)/2 & 0\end{pmatrix} B^*(y) A(y) P_+ \di x \di y
    \\
    & = F_{11}^+ 
\end{align*}
so we find
$$
g_\eps(\kappa) = \eps m U_{11}  + \eps^2 m F_{11}^{+} - \kappa + \eps^2\bigO(\kappa) +\bigO(\eps^3).
$$
We write $\kappa_0 = \eps m U_{11}  + \eps^2 m F_{11}^{+}$. 
Now, by using that $\kappa_0$ is of order at most $\eps$, we bound
$$
|g_\eps (\kappa) - (\kappa - \kappa_0)|  \le C_1 \eps^2 |\kappa_0| + C_1 \eps^2 |\kappa - \kappa_0| + C_2 \eps^3 \le C_3 \eps^3 + C_1 \eps^2 |\kappa - \kappa_0|.
$$
We restrict to $\eps$ sufficiently small such that $C_1 \eps^2 < 1/2$. 

Define the radius $r = 3 C_3 \epsilon^{3}$.
As soon as $\Re (\kappa_0) > 0$ and $\eps>0$ is sufficiently small, the ball 
$B(\kappa_0, r)$ is contained in $\Omega_+$
and we check that on $\partial B(\kappa_0, r)$,
we have
$$
|g_\eps (\kappa) - (\kappa - \kappa_0)| \le 
C_3 \eps^3 + C_1 \eps^2 r \le r/3 + r/2 < r.
$$
Hence, 
Rouché's theorem applies and shows the existence of a zero at distance $r = \bigO(\eps^3)$ from $\kappa_0$. Returning to the spectral parameter $z$, it has the expansion
$$
z_0  = \sqrt{m^2 - \kappa_0^2 + \kappa_0\bigO(\eps^3)}  = m -\frac{1}{2m} \kappa_0^2 +\kappa_0 \bigO(\eps^3) .
$$
On the other hand, if we take a contour approximating the boundary of a half-disc of radius $m$ centered at 0, we find that $g_\eps (\kappa)$ has exactly one zero inside this contour. Finally, for $|\kappa| > m$, there are no additional zeros. 
\end{proof}
 
 \begin{remark} \label{rmk:without-2nd-moment}
    As a final remark, let us outline what happens if we only assume $(1+|x|)V(x) \in L^1$. The only step in the proof where the integrability of $(1+|x|^2)V(x)$ was needed, is in \eqref{eq:snd-moment-bound}. By using that $M(z) - M^{(1)}(z)$ is a bounded operator (as a difference of bounded operators) and its kernel converges to zero pointwise, we obtain from dominated convergence
\begin{align*} 
   \lim_{\kappa \to 0} \left|\pscal{b,  \Par{M(z)- M^{(1)} }a }_{L^2}  \right| = 0.
\end{align*}
This gives the expansion
$$
g_\eps(\kappa) = \kappa_0 - \kappa +\eps^2\smallO(\kappa^0) +\bigO(\eps^3)
$$
and the corresponding eigenvalue expansion. 
 \end{remark}
 
\providecommand{\bysame}{\leavevmode\hbox to3em{\hrulefill}\thinspace}
\providecommand{\MR}{\relax\ifhmode\unskip\space\fi MR }
\providecommand{\MRhref}[2]{%
  \href{http://www.ams.org/mathscinet-getitem?mr=#1}{#2}
}
\providecommand{\href}[2]{#2}

\end{document}